\documentclass[twoside]{article}%
\usepackage{amssymb}
\usepackage{amsfonts}
\usepackage{amsmath}
\usepackage{graphicx}%
\providecommand{\U}[1]{\protect\rule{.1in}{.1in}}
\newtheorem{theorem}{Theorem}

\newtheorem{remark}[theorem]{Remark}

\newenvironment{proof}[1][Proof]{\noindent\textbf{#1.} }{\ \rule{0.5em}{0.5em}}
\begin{document}

\title{\textbf{Exact, Rotational, Infinite Energy, Blowup Solutions to the
3-Dimensional Euler Equations}}
\author{M\textsc{anwai Yuen\thanks{E-mail address: nevetsyuen@hotmail.com }}\\\textit{Department of Applied Mathematics, The Hong Kong Polytechnic
University,}\\\textit{Hung Hom, Kowloon, Hong Kong}}
\date{Revised 21-May-2011}
\maketitle

\begin{abstract}
In this paper, we construct a new class of blowup solutions with elementary
functions to the 3-dimensional compressible or incompressible Euler and
Navier-Stokes equations. In detail, we obtain a class of global rotational
exact solutions for the compressible fluids with $\gamma>1$:%
\begin{equation}
\left\{
\begin{array}
[c]{c}%
\rho=\max\left\{  \frac{\gamma-1}{K\gamma}\left[  C^{2}\left[  x^{2}%
+y^{2}+z^{2}-\left(  xy+yz+xz\right)  \right]  -\dot{a}(t)(x+y+z)+b(t)\right]
,0\right\}  ^{\frac{1}{\gamma-1}}\\
u_{1}=a(t)+C\left(  y-z\right)  \\
u_{2}=a(t)+C\left(  -x+z\right)  \\
u_{3}=a(t)+C\left(  x-y\right)
\end{array}
\right.
\end{equation}
where
\begin{equation}
a(t)=c_{0}+c_{1}t
\end{equation}
and
\begin{equation}
b(t)=3c_{0}c_{1}t+\frac{3}{2}c_{1}^{2}t^{2}+c_{2}%
\end{equation}
with $C$, $c_{0}$, $c_{1}$ and $c_{2}$ are arbitrary constants;\newline And
the corresponding blowup or global solutions for the incompressible Euler
equations are also given. Our constructed solutions are similar to the famous
Arnold-Beltrami-Childress (ABC) flow. The solutions with infinite energy can
exhibit the interesting behaviors locally. Besides, the corresponding global
solutions are also given for the compressible Euler equations. Furthermore,
due to $\operatorname{div}\vec{u}=0$ for the solutions, the solutions also
work for the 3-dimnsional incompressible Euler and Navier-Stokes equations.



\end{abstract}

\section{Introduction}

The $N$-dimensional Euler equations can be formulated as the follows:%
\begin{equation}
\left\{
\begin{array}
[c]{c}%
{\normalsize \rho}_{t}{\normalsize +\nabla\cdot(\rho\vec{u})}\text{
}{\normalsize =}\text{ }{\normalsize 0}\\
\rho\left[  \vec{u}_{t}+\left(  \vec{u}\cdot\nabla\right)  \vec{u}\right]
+\nabla P(\rho)\text{ }{\normalsize =}\text{ }0
\end{array}
\right.  \label{eq1}%
\end{equation}
where $\vec{x}=(x_{1},x_{2},...,x_{N})\in R^{N},$ $\rho=\rho(t,\vec{x})$ and
$\vec{u}=\vec{u}(t,\vec{x})=(u_{1},u_{2},...u_{n})\in R^{N}$ are the density
and the velocity respectively. The $\gamma$-law could be applied to the
pressure function, i.e.
\begin{equation}
P(\rho)=K\rho^{\gamma}\label{eq2}%
\end{equation}
with $K\geq0$ and $\gamma\geq1$. The Euler equations (\ref{eq1}) can be
rewritten by the scalar form,%
\begin{equation}
\left\{
\begin{array}
[c]{c}%
\frac{\partial\rho}{\partial t}+\sum_{k=1}^{N}u_{k}\frac{\partial\rho
}{\partial x_{k}}+\rho\sum_{k=1}^{N}\frac{\partial u_{k}}{\partial x_{k}%
}{\normalsize =}{\normalsize 0,}\\
\rho\left(  \frac{\partial u_{i}}{\partial t}+\sum_{k=1}^{N}u_{k}%
\frac{\partial u_{i}}{\partial x_{k}}\right)  +\frac{\partial P}{\partial
x_{i}}{\normalsize =0}\text{, for }i=1,2,...N.
\end{array}
\right.  \label{Tame}%
\end{equation}
The Euler equations (\ref{eq1}) are the fundamental model in fluid mechanics
\cite{CW} and \cite{L}.

Constructing exact solutions is a very important part in mathematical physics
to understand the nonlinear behaviors of the system. For the pressureless
fluids $K=0$, a class of exact blowup solutions were given in Yuen and Yeung's
papers \cite{YuenCNSNS2011a} and \cite{YYCNSNS2011}:%
\begin{equation}
\left\{
\begin{array}
[c]{c}%
\rho=\frac{f(\frac{x_{1}+d_{1}}{a_{1}(t)},\frac{x_{2}+d_{2}}{a_{2}%
(t)},....,\frac{x_{N}+d_{N}}{a_{N}(t)})}{\underset{i=1}{\overset{N}{\Pi}}%
a_{i}(t)}\text{, }u_{i}=\frac{\dot{a}_{i}(t)}{a_{i}(t)}\left(  x_{i}%
+d_{i}\right)  \text{ for }i=1,2,....,N\\
a_{i}(t)=a_{i0}+a_{i1}t
\end{array}
\right.
\end{equation}
with an arbitrary function $f\geq0$ and $f\in C^{1};$ and $d_{i}$, $a_{i0}>0$
and $a_{i1}$ are constants$.$\newline In particular, for $a_{2}<0$, the
solutions blow up in the finite time $T=-a_{1}/a_{2}$. There are other
analytical solutions for the compressible Euler or Navier-Stokes equations in
\cite{Sedov}, \cite{Ovsian}, \cite{Makino93exactsolutions}, \cite{Y},
\cite{Y2}, \cite{YuenCNSNS2011a}, \cite{YYCNSNS2011}, \cite{YuenPerEuler} and
\cite{YuenSelf-similarEuler1}. In addition, there are two papers to
investigate for a special structural form which generalizes the exact
solutions for Burgers' vortices in \cite{OG} and \cite{Consta}.

In this paper, we manipulate the elementary functions to construct some exact
rotational solutions for the 3-dimensional compressible Euler equations
(\ref{eq1}):

\begin{theorem}
For the 3-dimensional compressible Euler equations (\ref{eq1}), there exists a
class of rotational solutions:\newline for $\gamma>1$:%
\begin{equation}
\left\{
\begin{array}
[c]{c}%
\rho=\max\left\{  \frac{\gamma-1}{K\gamma}\left[  C^{2}\left[  x^{2}%
+y^{2}+z^{2}-\left(  xy+yz+xz\right)  \right]  -\dot{a}(t)(x+y+z)+b(t)\right]
,\text{ }0\right\}  ^{\frac{1}{\gamma-1}}\\
u_{1}=a(t)+C\left(  y-z\right) \\
u_{2}=a(t)+C\left(  -x+z\right) \\
u_{3}=a(t)+C\left(  x-y\right)
\end{array}
\right.  \label{YUEN}%
\end{equation}
where
\begin{equation}
a(t)=c_{0}+c_{1}t
\end{equation}
and
\begin{equation}
b(t)=3c_{0}c_{1}t+\frac{3}{2}c_{1}^{2}t^{2}+c_{2}%
\end{equation}
with $C$, $c_{0}$, $c_{1}$ and $c_{2}$ are arbitrary constants;\newline for
$\gamma=1$:%
\begin{equation}
\left\{
\begin{array}
[c]{c}%
\rho=e^{C^{2}\left[  x^{2}+y^{2}+z^{2}-\left(  xy+yz+xz\right)  \right]
-\dot{a}(t)(x+y+z)+b(t)}\\
u_{1}=a(t)+C\left(  y-z\right) \\
u_{2}=a(t)+C\left(  -x+z\right) \\
u_{3}=a(t)+C\left(  x-y\right)  .
\end{array}
\right.  \label{Yuen01}%
\end{equation}
The solutions (\ref{YUEN}) and (\ref{Yuen01}) globally exist.

\begin{remark}
In 1965, Arnold first introduced the famous Arnold-Beltrami-Childress (ABC)
flow%
\begin{equation}
\left\{
\begin{array}
[c]{c}%
u_{1}=A\sin z+C\cos y\\
u_{2}=B\sin x+A\cos z\\
u_{3}=C\sin y+B\cos x
\end{array}
\right.  \label{ABCFlows}%
\end{equation}
with constants $A$, $B$, $C$ and a suitable pressure function $P$ only for the
incompressible Euler equations in \cite{Arnold}. We observe that our solutions
(\ref{YUEN}) and (\ref{Yuen01}) are similar to the ABC flow.
\end{remark}

\begin{remark}
The solutions with infinite energy can exhibit the interesting behaviors
locally. The exact solutions with infinite energy of the systems may be
regionally applicable to understand the great complexity that exists in
turbulent phenomena.
\end{remark}

\begin{remark}
We notice that the rational functional form with $a(t)=0$ and $C=1$ in
solutions (\ref{YUEN}) and (\ref{Yuen01}) for the velocity $\vec{u}$ comes
from Senba and Suzuki's book \cite{ST}. The velocities $\vec{u}$ in solutions
(\ref{YUEN}) and (\ref{Yuen01}) are not spherically symmetric.
\end{remark}

\begin{remark}
The exact rotational solutions (\ref{YUEN}) and (\ref{Yuen01}) could be good
examples for testing numerical methods for fluid dynamics.
\end{remark}
\end{theorem}

\section{Compressible Rotational Fluids}

The main technique of this article is just to use the primary assumption about
the velocities $\vec{u}$:
\begin{equation}
\left\{
\begin{array}
[c]{c}%
u_{1}=a(t)+C\left(  y-z\right) \\
u_{2}=a(t)+C\left(  -x+z\right) \\
u_{3}=a(t)+C\left(  x-y\right)
\end{array}
\right.  \label{Yuenfunctional}%
\end{equation}
to substitute the governing equations (\ref{eq1}) to construct the density
function $\rho$ to balance the system. Therefore, the proof is simple to be
checked by direct computation:

\begin{proof}
For the vacuum solutions $\rho=0$, it is the trivial solutions for the system
(\ref{eq1}).\newline For $\gamma>1$, with non-vacuum solutions, we have the
first momentum equation (\ref{eq1})$_{2,1}$:%
\begin{equation}
\frac{\partial}{\partial t}u_{1}+u_{1}\frac{\partial}{\partial x}u_{1}%
+u_{2}\frac{\partial}{\partial y}u_{1}+u_{3}\frac{\partial}{\partial z}%
u_{1}+K\gamma\rho_{{}}^{\gamma-2}\frac{\partial}{\partial x}\rho
\end{equation}%
\begin{equation}
=\frac{\partial}{\partial t}u_{1}+u_{2}\frac{\partial}{\partial y}u_{1}%
+u_{3}\frac{\partial}{\partial z}u_{1}+\frac{K\gamma}{\gamma-1}\frac{\partial
}{\partial x}\rho_{{}}^{\gamma-1}%
\end{equation}
with%
\begin{equation}
\frac{K\gamma}{\gamma-1}\rho_{{}}^{\gamma-1}=C^{2}\left(  x^{2}+y^{2}%
+z^{2}-\left(  xy+yz+xz\right)  \right)  -\dot{a}(t)(x+y+z)+b(t)
\end{equation}
which we could require the condition for the density function:
\begin{equation}
\rho=\max\left\{  \frac{\gamma-1}{K\gamma}\left[  C^{2}\left[  x^{2}%
+y^{2}+z^{2}-\left(  xy+yz+xz\right)  \right]  -\dot{a}(t)(x+y+z)+b(t)\right]
,\text{ }0\right\}  ^{\frac{1}{\gamma-1}}, \label{Den}%
\end{equation}
to have%
\begin{align}
&  =\dot{a}(t)+\left[  a(t)+C(-x+z)\right]  \frac{\partial}{\partial y}\left[
a(t)+C\left(  y-z\right)  \right]  +\left[  a(t)+C\left(  x-y\right)  \right]
\frac{\partial}{\partial z}\left[  a(t)+C\left(  y-z\right)  \right]
\nonumber\\
&  +\frac{\partial}{\partial x}\left[  C^{2}\left[  x^{2}+y^{2}+z^{2}-\left(
xy+yz+xz\right)  \right]  -\dot{a}(t)(x+y+z)+b(t)\right]
\end{align}%
\begin{align}
&  =\dot{a}(t)+\left[  a(t)+C(-x+z)\right]  C+\left[  a(t)+C\left(
x-y\right)  (-C)\right] \nonumber\\
&  +\left[  2C^{2}x-C^{2}y-C^{2}z-\dot{a}(t)\right]
\end{align}%
\begin{equation}
=\dot{a}(t)-2C^{2}x+C^{2}z+C^{2}y+2C^{2}x-C^{2}y-C^{2}z-\dot{a}(t)
\end{equation}%
\begin{equation}
=0.
\end{equation}
\newline Similarly by the nice symmetry of the functions (\ref{YUEN}), we
could balance the second momentum equation (\ref{eq1})$_{2,2}$:%
\begin{equation}
\frac{\partial}{\partial t}u_{2}+u_{1}\frac{\partial}{\partial x}u_{2}%
+u_{3}\frac{\partial}{\partial z}u_{2}+\frac{K\gamma}{\gamma-1}\frac{\partial
}{\partial y}\rho_{{}}^{\gamma-1}%
\end{equation}%
\begin{align}
&  =\dot{a}(t)+\left[  a(t)+C\left(  y-z\right)  \right]  \frac{\partial
}{\partial x}\left[  a(t)+C\left(  -x+z\right)  \right]  +\left[
a(t)+C\left(  x-y\right)  \right]  \frac{\partial}{\partial z}\left[
a(t)+C\left(  -x+z\right)  \right] \nonumber\\
&  +\frac{\partial}{\partial y}\left[  C^{2}\left[  \left(  x^{2}+y^{2}%
+z^{2}\right)  -\left(  xy+yz+xz\right)  \right]  -\dot{a}%
(t)(x+y+z)+b(t)\right]
\end{align}%
\begin{align}
&  =\dot{a}(t)+\left[  a(t)+C\left(  y-z\right)  \right]  \left(  -C\right)
+\left[  a(t)+C\left(  x-y\right)  \right]  C\\
&  +2C^{2}y-C^{2}x-C^{2}z-\dot{a}(t)\nonumber
\end{align}%
\begin{equation}
=0.
\end{equation}
\newline For the last equation (\ref{eq1})$_{2,3}$, we have%
\begin{equation}
\frac{\partial}{\partial t}u_{3}+u_{1}\frac{\partial}{\partial x}u_{3}%
+u_{2}\frac{\partial}{\partial y}u_{3}+\frac{K\gamma}{\gamma-1}\frac{\partial
}{\partial z}\rho_{{}}^{\gamma-1}%
\end{equation}%
\begin{align}
&  =\dot{a}(t)+\left[  a(t)+C\left(  y-z\right)  \right]  \frac{\partial
}{\partial x}\left[  a(t)+C\left(  x-y\right)  \right]  +\left[  a(t)+C\left(
-x+z\right)  \right]  \frac{\partial}{\partial y}\left[  a(t)+C\left(
x-y\right)  \right] \nonumber\\
&  +\frac{\partial}{\partial z}\left[  C^{2}\left[  \left(  x^{2}+y^{2}%
+z^{2}\right)  -\left(  xy+yz+xz\right)  \right]  -\dot{a}%
(t)(x+y+z)+b(t)\right]
\end{align}%
\begin{equation}
=\dot{a}(t)+\left[  a(t)+C\left(  y-z\right)  \right]  C+\left[  a(t)+C\left(
-x+z\right)  \right]  (-C)+2C^{2}z-C^{2}y-C^{2}x-\dot{a}(t)
\end{equation}%
\begin{equation}
=0.
\end{equation}

For the mass equation (\ref{eq1})$_{1}$, we have with the density (\ref{Den})
to obtain:%
\begin{equation}
\rho_{t}+\nabla\cdot(\rho\vec{u})
\end{equation}%
\begin{align}
&  =\frac{\partial}{\partial t}\left[  \frac{\gamma-1}{K\gamma}\left(
C^{2}\left(  x^{2}+y^{2}+z^{2}-\left(  xy+yz+xz\right)  \right)  -\dot
{a}(t)(x+y+z)+b(t)\right)  \right]  ^{\frac{1}{\gamma-1}}\nonumber\\
&  +\frac{\partial}{\partial x}\left[  \frac{\gamma-1}{K\gamma}\left(
C^{2}\left(  x^{2}+y^{2}+z^{2}-\left(  xy+yz+xz\right)  \right)  -\dot
{a}(t)(x+y+z)+b(t)\right)  \right]  ^{\frac{1}{\gamma-1}}\left(
a(t)+C(y-z)\right) \nonumber\\
&  +\frac{\partial}{\partial y}\left[  \frac{\gamma-1}{K\gamma}\left(
C^{2}\left(  x^{2}+y^{2}+z^{2}-\left(  xy+yz+xz\right)  \right)  -\dot
{a}(t)(x+y+z)+b(t)\right)  \right]  ^{\frac{1}{\gamma-1}}\left(  a(t)+C\left(
-x+z\right)  \right) \nonumber\\
&  +\frac{\partial}{\partial z}\left[  \frac{\gamma-1}{K\gamma}\left(
C^{2}\left(  x^{2}+y^{2}+z^{2}-\left(  xy+yz+xz\right)  \right)  -\dot
{a}(t)(x+y+z)+b(t)\right)  \right]  ^{\frac{1}{\gamma-1}}\left(  a(t)+C\left(
x-y\right)  \right)
\end{align}%
\begin{align}
&  =\frac{1}{\gamma-1}\left[  \frac{\gamma-1}{K\gamma}\left(  C^{2}\left(
x^{2}+y^{2}+z^{2}-\left(  xy+yz+xz\right)  \right)  -\dot{a}%
(t)(x+y+z)+b(t)\right)  \right]  ^{\frac{1}{\gamma-1}-1}\nonumber\\
&  \cdot\left\{
\begin{array}
[c]{c}%
\frac{\gamma-1}{K\gamma}\left[  -\ddot{a}(t)(x+y+z)+\dot{b}(t)\right] \\
+\frac{\gamma-1}{K\gamma}\left[  \left(  2C^{2}x-C^{2}y-C^{2}z-\dot
{a}(t)\right)  \left(  a(t)+C(y-z)\right)  \right] \\
+\frac{\gamma-1}{K\gamma}\left[  \left(  2C^{2}y-C^{2}x-C^{2}z-\dot
{a}(t)\right)  \left(  a(t)+C\left(  -x+z\right)  \right)  \right] \\
+\frac{\gamma-1}{K\gamma}\left[  \left(  2C^{2}z-C^{2}x-C^{2}y-\dot
{a}(t)\right)  \left(  a(t)+C(x-y)\right)  \right]
\end{array}
\right\}
\end{align}%
\begin{align}
&  =\frac{1}{K\gamma}\left[  \frac{\gamma-1}{K\gamma}\left(  C^{2}\left(
x^{2}+y^{2}+z^{2}-\left(  xy+yz+xz\right)  \right)  -\dot{a}%
(t)(x+y+z)+b(t)\right)  \right]  ^{\frac{1}{\gamma-1}-1}\nonumber\\
&  \cdot\left\{
\begin{array}
[c]{c}%
-\ddot{a}(t)(x+y+z)+\dot{b}(t)\\
+\left(  2C^{2}x-C^{2}y-C^{2}z-\dot{a}(t)\right)  a(t)+C\left(  2C^{2}%
x-C^{2}y-C^{2}z-\dot{a}(t)\right)  (y-z)\\
+\left(  2C^{2}y-C^{2}x-C^{2}z-\dot{a}(t)\right)  a(t)+C\left(  2C^{2}%
y-C^{2}x-C^{2}z-\dot{a}(t)\right)  \left(  -x+z\right) \\
+\left(  2C^{2}z-C^{2}x-C^{2}y-\dot{a}(t)\right)  a(t)+C\left(  2C^{2}%
z-C^{2}x-C^{2}y-\dot{a}(t)\right)  (x-y)
\end{array}
\right\}
\end{align}%
\begin{align}
&  =\frac{1}{K\gamma}\left[  \frac{\gamma-1}{K\gamma}\left(  C^{2}\left(
x^{2}+y^{2}+z^{2}-\left(  xy+yz+xz\right)  \right)  -\dot{a}%
(t)(x+y+z)+b(t)\right)  \right]  ^{\frac{1}{\gamma-1}-1}\nonumber\\
&  \cdot\left\{  -\ddot{a}(t)(x+y+z)+\dot{b}(t)-3\dot{a}(t)a(t)\right\}
\end{align}%
\begin{equation}
=0
\end{equation}
by choosing%
\begin{equation}
a(t)=c_{0}+c_{1}t
\end{equation}
and%
\begin{equation}
\dot{b}(t)=3\dot{a}(t)a(t)
\end{equation}%
\begin{equation}
\dot{b}(t)=3c_{1}\left(  c_{0}+c_{1}t\right)
\end{equation}%
\begin{equation}
b(t)=3c_{0}c_{1}t+\frac{3}{2}c_{1}^{2}t^{2}+c_{2}%
\end{equation}
where $c_{2}$ is the arbitrary constant.\newline As for $\gamma=1$, the proof
is similar, we may skip the details here.\newline It is clear to see that the
solutions (\ref{YUEN}) and (\ref{Yuen01}) globally exist.\newline The proof is completed.
\end{proof}

Here the masses of the solutions (\ref{YUEN}) and (\ref{Yuen01}) are infinite:%
\begin{equation}
\int_{R^{3}}\rho dx=+\infty.
\end{equation}

\section{Incompressible Rotational Fluids}

For the 3-dimensional incompressible Euler equations:%
\begin{equation}
\left\{
\begin{array}
[c]{c}%
\operatorname{div}\vec{u}\text{ }=\text{ }0\\
\rho\left[  \vec{u}_{t}+\left(  \vec{u}\cdot\nabla\right)  \vec{u}\right]
+\nabla P(\rho)\text{ }=\text{ }0,
\end{array}
\right.  \label{incom}%
\end{equation}
the solutions with elliptical symmetric velocity in \cite{Sedov},
\cite{Ovsian} and \cite{YuenSelf-similarEuler1} for the compressible Euler
equations (\ref{eq1}), could not be applied to the incompressible Euler
equations (\ref{incom}) because of the linear functional form:%
\begin{equation}
\left\{
\begin{array}
[c]{c}%
u_{1}=c_{1}(t)x+d_{1}(t)\\
u_{2}=c_{2}(t)y+d_{2}(t)\\
u_{3}=c_{3}(t)z+d_{3}(t)
\end{array}
\right.  \label{Lii11}%
\end{equation}
with some functions $c_{i}(t)$ and $d_{i}(t)$ to have%
\begin{equation}
\operatorname{div}\vec{u}\neq0
\end{equation}
for non-trivial solutions $c_{1}(t)+c_{2}(t)+c_{3}(t)\neq0$. However, our
solutions (\ref{YUEN}) and (\ref{Yuen01}) could be applied to the
incompressible Euler equations (\ref{incom}). The set of the corresponding
solutions for the incompressible Euler equations (\ref{incom}) could be larger
than the compressible ones (\ref{eq1}):

\begin{theorem}
For the 3-dimensional incompressible Euler equations (\ref{incom}), there
exists a class of rotational solutions:%
\begin{equation}
\left\{
\begin{array}
[c]{c}%
\frac{P(\rho)}{\rho}=\left[
\begin{array}
[c]{c}%
C^{2}\left[  x^{2}+y^{2}+z^{2}-\left(  xy+yz+xz\right)  \right]  -\left[
\dot{a}_{1}(t)+C(a_{2}(t)+a_{3}(t))\right]  x\\
-\left[  \dot{a}_{2}(t)+C(a_{1}(t)+a_{3}(t))\right]  y-\left[  \dot{a}%
_{3}(t)+C(a_{1}(t)+a_{2}(t))\right]  z+b(t)
\end{array}
\right] \\
u_{1}=a_{1}(t)+C\left(  y-z\right) \\
u_{2}=a_{2}(t)+C\left(  -x+z\right) \\
u_{3}=a_{3}(t)+C\left(  x-y\right)
\end{array}
\right.  \label{YUEN2}%
\end{equation}
and
\begin{equation}
\left\{
\begin{array}
[c]{c}%
\frac{P(\rho)}{\rho}=\left[
\begin{array}
[c]{c}%
-C^{2}\left(  x^{2}+y^{2}+z^{2}+xy+yz+xz\right) \\
-\left[  \dot{a}_{1}(t)+C(a_{2}(t)+a_{3}(t))\right]  x\\
-\left[  \dot{a}_{2}(t)+C(a_{1}(t)+a_{3}(t))\right]  y\\
-\left[  \dot{a}_{3}(t)+C(a_{1}(t)+a_{2}(t))\right]  z+b(t)
\end{array}
\right] \\
u_{1}=a_{1}(t)+C\left(  y+z\right) \\
u_{2}=a_{2}(t)+C\left(  x+z\right) \\
u_{3}=a_{3}(t)+C\left(  x+y\right)
\end{array}
\right.  \label{Yuen3}%
\end{equation}
where $a_{i}(t)$ is an arbitrary local $C^{1}$ function for $i=1,$ $2$ $or$
$3,$ $b(t)$ is an arbitrary function, $C$ is an arbitrary constant.\newline In
particular,\newline1. if $\left\vert a_{i}(T)\right\vert =\infty$ or
$\left\vert \dot{a}_{i}(T)\right\vert =\infty$ with the first finite positive
constant $T$, the solutions (\ref{YUEN2}) and (\ref{Yuen3}) blow up in the
finite time $T$;\newline2. For global $C^{1}$ functions $a_{i}(t)$, the
solutions (\ref{YUEN2}) and (\ref{Yuen3}) globally exist.
\end{theorem}

\begin{proof}
For the checking of the functions (\ref{YUEN2}), the detail is similar to the
proof of Theorem 1:

For the checking of the functions (\ref{Yuen3}), we have for the first
momentum equation (\ref{incom})$_{2,2}$:%
\begin{equation}
\frac{\partial}{\partial t}u_{1}+u_{1}\frac{\partial}{\partial x}u_{1}%
+u_{2}\frac{\partial}{\partial y}u_{1}+u_{3}\frac{\partial}{\partial z}%
u_{1}+\frac{\partial}{\partial x}\left(  \frac{P}{\rho}\right)
\end{equation}%
\begin{equation}
=\frac{\partial}{\partial t}u_{1}+u_{2}\frac{\partial}{\partial y}u_{1}%
+u_{3}\frac{\partial}{\partial z}u_{1}+\frac{\partial}{\partial x}\left(
\frac{P}{\rho}\right)
\end{equation}
by requiring the pressure function:
\begin{equation}
\frac{P}{\rho}=\left[
\begin{array}
[c]{c}%
C^{2}\left[  x^{2}+y^{2}+z^{2}-\left(  xy+yz+xz\right)  \right]  -\left[
\dot{a}_{1}(t)+C(a_{2}(t)+a_{3}(t))\right]  x\\
-\left[  \dot{a}_{2}(t)+C(a_{1}(t)+a_{3}(t))\right]  y-\left[  \dot{a}%
_{3}(t)+C(a_{1}(t)+a_{2}(t))\right]  z+b(t)
\end{array}
\right]
\end{equation}
to have%
\begin{align}
&  =\dot{a}_{1}(t)+\left[  a_{2}(t)+C(-x+z)\right]  \frac{\partial}{\partial
y}\left[  a_{1}(t)+C\left(  y-z\right)  \right]  +\left[  a_{3}(t)+C\left(
x-y\right)  \right]  \frac{\partial}{\partial z}\left[  a_{1}(t)+C\left(
y-z\right)  \right] \nonumber\\
&  +\frac{\partial}{\partial x}\left[
\begin{array}
[c]{c}%
C^{2}\left[  x^{2}+y^{2}+z^{2}-\left(  xy+yz+xz\right)  \right]  -\left[
\dot{a}_{1}(t)+C(a_{2}(t)+a_{3}(t))\right]  x\\
-\left[  \dot{a}_{2}(t)+C(a_{1}(t)+a_{3}(t))\right]  y-\left[  \dot{a}%
_{3}(t)+C(a_{1}(t)+a_{2}(t))\right]  z+b(t)
\end{array}
\right]
\end{align}%
\begin{align}
&  =\dot{a}_{1}(t)+\left[  a_{2}(t)+C(-x+z)\right]  C+\left[  a_{3}%
(t)+C\left(  x-y\right)  (-C)\right] \nonumber\\
&  +\left[  2C^{2}x-C^{2}y-C^{2}z-\dot{a}(t)-C(a_{2}(t)+a_{3}(t))\right]
\end{align}%
\begin{equation}
=\dot{a}_{1}(t)+C(a_{2}(t)+a_{3}(t))-2C^{2}x+C^{2}z+C^{2}y+2C^{2}%
x-C^{2}y-C^{2}z-\dot{a}(t)-C(a_{2}(t)+a_{3}(t))
\end{equation}%
\begin{equation}
=0.
\end{equation}
\newline Similarly by the nice symmetry of the functions (\ref{YUEN2}), we
could balance the second momentum equation (\ref{incom})$_{2,2}$:%
\begin{equation}
\frac{\partial}{\partial t}u_{2}+u_{1}\frac{\partial}{\partial x}u_{2}%
+u_{3}\frac{\partial}{\partial z}u_{2}+\frac{\partial}{\partial y}\left(
\frac{P}{\rho}\right)
\end{equation}%
\begin{align}
&  =\dot{a}_{2}(t)+\left[  a_{1}(t)+C\left(  y-z\right)  \right]
\frac{\partial}{\partial x}\left[  a_{2}(t)+C\left(  -x+z\right)  \right]
+\left[  a_{3}(t)+C\left(  x-y\right)  \right]  \frac{\partial}{\partial
z}\left[  a_{2}(t)+C\left(  -x+z\right)  \right] \nonumber\\
&  +\frac{\partial}{\partial y}\left[
\begin{array}
[c]{c}%
C^{2}\left[  x^{2}+y^{2}+z^{2}-\left(  xy+yz+xz\right)  \right]  -\left[
\dot{a}_{1}(t)+C(a_{2}(t)+a_{3}(t))\right]  x\\
-\left[  \dot{a}_{2}(t)+C(a_{1}(t)+a_{3}(t))\right]  y-\left[  \dot{a}%
_{3}(t)+C(a_{1}(t)+a_{2}(t))\right]  z+b(t)
\end{array}
\right]
\end{align}%
\begin{align}
&  =\dot{a}_{2}(t)+\left[  a_{1}(t)+C\left(  y-z\right)  \right]  \left(
-C\right)  +\left[  a_{3}(t)+C\left(  x-y\right)  \right]  C\\
&  +2C^{2}y-C^{2}x-C^{2}z-\dot{a}_{2}(t)-C(a_{1}(t)+a_{3}(t))\nonumber
\end{align}%
\begin{equation}
=0.
\end{equation}
\newline For the last equation (\ref{incom})$_{2,3}$, we have%
\begin{equation}
\frac{\partial}{\partial t}u_{3}+u_{1}\frac{\partial}{\partial x}u_{3}%
+u_{2}\frac{\partial}{\partial y}u_{3}+\frac{\partial}{\partial z}\left(
\frac{P}{\rho}\right)
\end{equation}%
\begin{align}
&  =\dot{a}_{3}(t)+\left[  a_{1}(t)+C\left(  y-z\right)  \right]
\frac{\partial}{\partial x}\left[  a_{3}(t)+C\left(  x-y\right)  \right]
+\left[  a_{2}(t)+C\left(  -x+z\right)  \right]  \frac{\partial}{\partial
y}\left[  a_{3}(t)+C\left(  x-y\right)  \right] \nonumber\\
&  +\frac{\partial}{\partial z}\left[
\begin{array}
[c]{c}%
C^{2}\left[  x^{2}+y^{2}+z^{2}-\left(  xy+yz+xz\right)  \right]  -\left[
\dot{a}_{1}(t)+C(a_{2}(t)+a_{3}(t))\right]  x\\
-\left[  \dot{a}_{2}(t)+C(a_{1}(t)+a_{3}(t))\right]  y-\left[  \dot{a}%
_{3}(t)+C(a_{1}(t)+a_{2}(t))\right]  z+b(t)
\end{array}
\right]
\end{align}%
\begin{equation}
=\dot{a}_{3}(t)+\left[  a_{1}(t)+C\left(  y-z\right)  \right]  C+\left[
a_{2}(t)+C\left(  -x+z\right)  \right]  (-C)+2C^{2}z-C^{2}y-C^{2}x-\dot{a}%
_{3}(t)-C(a_{1}(t)+a_{2}(t))
\end{equation}%
\begin{equation}
=0.
\end{equation}

Next, for the checking of the solutions (\ref{Yuen3}), we have for the first
momentum equation (\ref{incom})$_{2,1}$:%
\begin{equation}
=\frac{\partial}{\partial t}u_{1}+u_{2}\frac{\partial}{\partial y}u_{1}%
+u_{3}\frac{\partial}{\partial z}u_{1}+\frac{\partial}{\partial x}\left(
\frac{P}{\rho}\right)
\end{equation}
with the pressure function%
\begin{equation}
\frac{P}{\rho}=\left[
\begin{array}
[c]{c}%
-C^{2}\left(  x^{2}+y^{2}+z^{2}+xy+yz+xz\right)  -\left[  \dot{a}%
_{1}(t)+C(a_{2}(t)+a_{3}(t))\right]  x\\
-\left[  \dot{a}_{2}(t)+C(a_{1}(t)+a_{3}(t))\right]  y-\left[  \dot{a}%
_{3}(t)+C(a_{1}(t)+a_{2}(t))\right]  z+b(t)
\end{array}
\right]
\end{equation}
to have%
\begin{align}
&  =\dot{a}_{1}(t)+\left[  a_{2}(t)+C(x+z)\right]  \frac{\partial}{\partial
y}\left[  a_{1}(t)+C\left(  y+z\right)  \right]  +\left[  a_{3}(t)+C\left(
x+y\right)  \right]  \frac{\partial}{\partial z}\left[  a_{1}(t)+C\left(
y+z\right)  \right] \nonumber\\
&  +\frac{\partial}{\partial x}\left[
\begin{array}
[c]{c}%
-C^{2}\left(  x^{2}+y^{2}+z^{2}+xy+yz+xz\right)  -\left[  \dot{a}%
_{1}(t)+C(a_{2}(t)+a_{3}(t))\right]  x\\
-\left[  \dot{a}_{2}(t)+C(a_{1}(t)+a_{3}(t))\right]  y-\left[  \dot{a}%
_{3}(t)+C(a_{1}(t)+a_{2}(t))\right]  z+b(t)
\end{array}
\right]
\end{align}%
\begin{equation}
=\dot{a}_{1}(t)+\left[  a_{2}(t)+C(x+z)\right]  C+\left[  a_{3}(t)+C\left(
x+y\right)  \right]  C-2C^{2}x-C^{2}y-C^{2}z-\dot{a}(t)-C(a_{2}(t)+a_{3}(t))
\end{equation}%
\begin{equation}
=\dot{a}_{1}(t)+C(a_{2}(t)+a_{3}(t))+2C^{2}x+C^{2}z+C^{2}y-2C^{2}%
x-C^{2}y-C^{2}z-\dot{a}_{1}(t)-C(a_{2}(t)+a_{3}(t))
\end{equation}%
\begin{equation}
=0.
\end{equation}
Similarly, we can balance the second momentum equation (\ref{incom})$_{2,2}$
with the symmetric form of the functions (\ref{Yuen3}):%
\begin{equation}
=\frac{\partial}{\partial t}u_{2}+u_{1}\frac{\partial}{\partial x}u_{2}%
+u_{3}\frac{\partial}{\partial z}u_{2}+\frac{\partial}{\partial y}\left(
\frac{P}{\rho}\right)
\end{equation}%
\begin{equation}
=\dot{a}_{2}(t)+\left[  a_{1}(t)+C\left(  y+z\right)  \right]  \frac{\partial
}{\partial x}\left[  a_{2}(t)+C\left(  x+z\right)  \right]  +\left[
a_{3}(t)+C\left(  x+y\right)  \right]  \frac{\partial}{\partial z}\left[
a_{2}(t)+C\left(  x+z\right)  \right]  +\frac{\partial}{\partial y}\left(
\frac{P}{\rho}\right)
\end{equation}%
\begin{equation}
=\dot{a}_{2}(t)+\left[  a_{1}(t)+C\left(  y+z\right)  \right]  C+\left[
a_{3}(t)+C\left(  x+y\right)  \right]  C-2C^{2}y-C^{2}x-C^{2}z-\dot{a}%
_{2}(t)-C(a_{1}(t)+a_{3}(t))\nonumber
\end{equation}%
\begin{equation}
=0.
\end{equation}
For the last equation (\ref{incom})$_{2,3}$, we could get%
\begin{equation}
=\frac{\partial}{\partial t}u_{3}+u_{1}\frac{\partial}{\partial x}u_{3}%
+u_{2}\frac{\partial}{\partial y}u_{3}+\frac{\partial}{\partial z}\left(
\frac{P}{\rho}\right)
\end{equation}%
\begin{equation}
=\dot{a}_{3}(t)+\left[  a_{1}(t)+C\left(  y+z\right)  \right]  \frac{\partial
}{\partial x}\left[  a_{3}(t)+C\left(  x+y\right)  \right]  +\left[
a_{2}(t)+C\left(  x+z\right)  \right]  \frac{\partial}{\partial y}\left[
a_{3}(t)+C\left(  x+y\right)  \right]  +\frac{\partial}{\partial z}\left(
\frac{P}{\rho}\right) \nonumber
\end{equation}%
\begin{equation}
=\dot{a}_{3}(t)+\left[  a_{1}(t)+C\left(  y+z\right)  \right]  C+\left[
a_{2}(t)+C\left(  x+z\right)  \right]  C-2C^{2}z-C^{2}y-C^{2}x-\dot{a}%
_{3}(t)-C(a_{1}(t)+a_{2}(t))
\end{equation}%
\begin{equation}
=0.
\end{equation}
It is clear to see that\newline1. if $\left\vert a_{i}(T)\right\vert =\infty$
or $\left\vert \dot{a}_{i}(T)\right\vert =\infty$ for $i=1,$ $2$ or 3,with the
first finite positive constant $T$, the solutions (\ref{YUEN2}) and
(\ref{Yuen3}) blow up in the finite time $T$;\newline2. for all global $C^{1}$
functions $a_{i}(t)$, the solutions (\ref{YUEN2}) and (\ref{Yuen3}) globally
exist.\newline The proof is completed.
\end{proof}

Here the kinetic energy of the solutions (\ref{YUEN2}) and (\ref{Yuen3}) for
the incompressible fluids are infinite:%
\begin{equation}
\frac{1}{2}\int_{R^{3}}\vec{u}^{2}=+\infty.
\end{equation}
And the blowup solutions (\ref{YUEN2}) and (\ref{Yuen3}) are the other
examples with infinite energy for the incompressible Euler equations
(\ref{incom}), with respect to the blowup cylindrical ones in Gibbon, Moore
and Stuart's work \cite{GMS}.

Additionally it is because
\begin{equation}
\Delta\vec{u}=\vec{0}%
\end{equation}
in the solutions (\ref{YUEN}) and (\ref{Yuen01}) to be the corresponding
solutions for the 3-dimensional compressible Navier-Stokes equations:%
\begin{equation}
\left\{
\begin{array}
[c]{c}%
{\normalsize \rho}_{t}{\normalsize +\nabla\cdot(\rho\vec{u})}\text{
}{\normalsize =}\text{ }{\normalsize 0}\\
\rho\left[  \vec{u}_{t}+\left(  \vec{u}\cdot\nabla\right)  \vec{u}\right]
+\nabla P(\rho)\text{ }{\normalsize =}\text{ }\mu\Delta\vec{u}%
\end{array}
\right.
\end{equation}
with $\mu>0.$\newline And the solutions (\ref{YUEN2}) and (\ref{Yuen3}) are
the corresponding ones for the 3-dimensional incompressible Navier-Stokes
equations:%
\begin{equation}
\left\{
\begin{array}
[c]{c}%
{\normalsize \operatorname{div}\vec{u}}\text{ }{\normalsize =}\text{
}{\normalsize 0}\\
\rho\left[  \vec{u}_{t}+\left(  \vec{u}\cdot\nabla\right)  \vec{u}\right]
+\nabla P(\rho)\text{ }{\normalsize =}\text{ }\mu\Delta\vec{u}%
\end{array}
\right.
\end{equation}
with $\mu>0.$

\section{Discussion}

1. We observe that the constructed density functions $\rho$ (\ref{YUEN}) and
(\ref{Yuen01}), and the pressure functions (\ref{YUEN2}) and (\ref{Yuen3})
share some nice algebraic geometry, but we do not understand for the
implication of this geometric structure for the fluids. Could the solutions
with small perturbations of these solutions converge to the original ones?

2. Could we just cut a finite partial volume of the solutions to construct new
week solutions? If we could, what happens for the corresponding weak solutions
of the classical ones with infinite energy?

3. By taking consideration with Makino's blowup solutions form (\ref{Lii11}),
we could guess that the more general functional form for the $3$-dimensional
velocities $\vec{u}$ in Cartesian coordinate:%
\begin{equation}
\left\{
\begin{array}
[c]{c}%
u_{1}=a_{1}(t)+b_{11}(t)x+b_{12}(t)y+b_{13}(t)z\\
u_{2}=a_{2}(t)+b_{21}(t)x+b_{22}(t)y+b_{23}(t)z\\
u_{3}=a_{3}(t)+b_{31}(t)x+b_{32}(t)y+b_{33}(t)z
\end{array}
\right.  \label{functional}%
\end{equation}
with the local smooth time $C^{1}$ functions $a_{i}(t)$ and $b_{ij}(t)$ for
$i,$ $j=1,$ $2$ and $3,$\newline could be the more general solutions structure
for the systems. Then, it is more difficult to handle the linear functional
structure for the higher dimensional cases $R^{N}$ with $N\geq4$. Could we
construct a blowup example with rotation for the incompressible fluids? In
principle, we could adopt the guessing approach to construct the density
functions. Some trials with luck are needed to obtain a right novel functional
form to construct solutions for the systems, but it costs a lot of time for
computation. Thus, we hope to arouse the readers' interests to investigate
more systematic approaches to calculate the possible particular solutions with
the linear functional velocities $\vec{u}$ (\ref{functional}) in Cartesian
coordinate or other coordinates.

Finally, further researches are needed for understanding this class
(\ref{functional}) of blowup or global solutions for the systems in the future.

\end{document}